\documentclass[conference]{IEEEtran}
\usepackage{amsmath,amsthm,amsfonts,amscd,amssymb,graphicx,graphics,verbatim}
\usepackage{eucal}     
\usepackage{verbatim}       
\usepackage{makeidx}        
\usepackage{epsfig}             
\usepackage{graphicx,mathrsfs}
\usepackage{url,tikz}
\usetikzlibrary{arrows,shapes}

\newcommand{\eq}[1]{\begin{equation}#1 \end{equation}}

\newcommand{\Ex}[1]{\mathbb{E}\left[#1\right]}
\newcommand{\korner}{K\"{o}rner }
\newcommand{\entropy}[1]{\frac{1}{2}\log\left(2 \pi e #1\right)}
\newcommand{\entr}[2]{\frac{#1}{2}\log\left( #2\right)}

\newtheorem{thm}{Theorem}
\newtheorem{rem}{Remark}
\newtheorem{lem}{Lemma}
\newtheorem{cor}{Corollary}
\newtheorem{defin}{Definition}

\author{
\IEEEauthorblockN{Amin Jafarian}
\IEEEauthorblockA{Lab. of Informatics, Networks \& Communications (LINC)\\Department of Electrical \& Computer Engineering\\
University of Texas at Austin \\
Austin, TX 78712, USA\\
Email: jafarian@mail.utexas.edu}
\and
\authorblockN{Sriram Vishwanath}
\authorblockA{Lab. of Informatics, Networks  \& Communications (LINC)\\
Department of Electrical \& Computer Engineering\\
University of Texas at Austin\\
Austin, TX 78712, USA\\
Email: sriram@ece.utexas.edu}
}

\title{The Two-User Gaussian Fading Broadcast Channel}
\date{ }
\begin{document}
\maketitle
\begin{abstract}
This paper presents outerbounds for the two-user Gaussian fading broadcast channel. These outerbounds are based on Costa's entropy power inequality (Costa-EPI) and are formulated mathematically as a feasibility problem. For classes of the two-user Gaussian fading broadcast channel where the outerbound is found to have a feasible solution, we find conditions under which a suitable inner and outer bound meet. For all such cases, this paper provides a {\em partial} characterization of the capacity region of the Gaussian two-user fading broadcast channel.
\end{abstract}

\section{Introduction}
The Gaussian fading broadcast channel is of the basic Gaussian channels whose capacity region  still remains unknown. This channel is a broadcast channel with additive Gaussian noise and multiplicative state, as depicted in Figure \ref{fig:chmodel}. This multiplicative state is unknown to the transmitter while being known to the receivers. Such a channel represents one of the simplest models for down-link communication in a conventional cellular system where there is no channel-state feedback to the transmitter. This channel is different from most other Gaussian broadcast channels analyzed in literature in that it is, in general, a non-degraded (and non-more-capable) broadcast channel. This channel has received considerable attention in recent years \cite{JV09,TYL08}, with inner and outer bounds presented for this channel.

\begin{figure}[h!]
\begin{center}
\includegraphics[width=.4\textwidth]{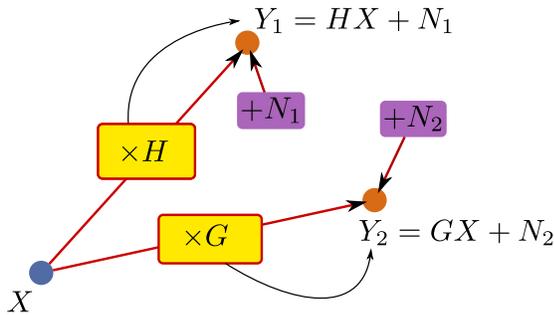}
\caption{Gaussian Fading Broadcast Channel}
\end{center}
\label{fig:chmodel}
\end{figure}

Gaussian fading broadcast channels have been studied and characterized in many other contexts, including when the channel state is known to both the transmitter and receivers. In this setting, it is found to reduce to multiple parallel degraded broadcast channels \cite{GL01}. In the case of vector Gaussian fading broadcast channels with state known to all parties, again, we have a separation into parallel channels, where dirty paper coding yields the rate region for each parallel case. However, such a separation does not hold in the case where the state is not known to the transmitter. Moreover, dirty paper coding as employed in \cite{JG04} does not directly apply to such settings.

In order to improving our understanding of the capacity limits of this channel, we develop a class of outer bounds for this channel based on Costa's entropy power inequality (Costa-EPI). Conventional EPI is typically effective for analyzing degraded (and parallel) broadcast channels, but is typically found not directly applicable for non-degraded cases. Costa-EPI provides us with {\em constraints} that enable us to develop an optimization framework for an outer bound on the capacity region of the fading broadcast channel. When this optimization problem has a feasible solution, it yields a potentially tighter outer bound on the rate region that its genie-based (degraded) counterparts \cite{TYL08,JV09, ALG05}.  In some special cases, the rates obtained using this outer bound can be achieved, thus providing a partial characterization of the channel's capacity region. 

In summary, our main results in this paper are as follows:

\begin{itemize}
\item We use Costa-EPI to develop an optimization-based framework for an outer bound for this channel (Theorem 1).  In general, this optimization problem is a feasibility problem. When feasible, it presents a bound that is in general tighter that other genie aided approaches.
\item For certain classes of channels, we show that the outer bound is in fact achievable, thus characterizing a portion of the capacity region of this channel.
\end{itemize}

The rest of this paper is organized as follows: the next section presents the system model. A description of the background on outer bounds for the broadcast channel and the particular outer bounding framework developed in this paper is presented in Section III. Section IV presents the achievable scheme and compares it with the outer bound.

\section{Channel Model}
\label{sec:model}

The two user broadcast fading channel is given mathematically by $Y_1=X+\frac{N_1}{H}$
and $Y_2=X+\frac{N_2}{G}$.

where $X$ denotes the output of the Transmitter which is limited  to a power of $Q$, and $Y_i$ corresponds to the input observed by Receiver $i$, $i \in \{1,2\}$. Each received signal undergoes a fade corresponding to random variables $H$ for Receiver 1 and $G$ for Receiver 2. This fade is assumed to be known perfectly to each receiver, and is thus modeled as a multiplicative factor impacting the receiver-side additive Gaussian noise $N_i$ at each receiver. This received side noise is assumed to be normalized to be of unit variance at each receiver. As stated in the introduction, the instantiations of neither $H$ nor $G$ are known to the transmitter. The transmitter, however, is aware of the distributions of $H$ and $G$. Our goal is to find non-trivial outer bounds on the capacity region of this channel, and when possible, exact sum-capacity results for this channel.

For simplicity, both $H$ and $G$ are assumed to have a discrete support. Thus, $H$ can take one of $n$ possible fade-states $h_1,\dots,h_n$ in accordance with the p.m.f $[p_1,p_2,\ldots,p_n]$. Similarly, $G$ takes on one of $m$ possible values $g_1,\dots,g_m$ according to the p.m.f. $[q_1,q_2,\ldots,q_m]$ .

\section{Background \& Main Result}
In this section, we first present an overview on exiting literature regarding outer bounds for the broadcast channel. Second, we present the main result based on an optimization framework for the outer bound based on Costa-EPI.

\subsection{Definitions and Background}

We define $C_1$ to equal the ergodic point-to-point capacity of Receiver 1. In other words, if Receiver 2 were absent from the channel model, the ergodic capacity is known and equals:
\[
C_1 \triangleq \frac{1}{2} \Ex{ \log ( 1 + Q H^2)}
\]

In the same spirit, $C_2$ is defined to equal 
\[C_2\triangleq \frac{1}{2} \Ex{ \log ( 1 + Q G^2)}\]

The first outer bound for the broadcast channel, developed from the expressions in \cite{ahlswede74_2}, is fairly intuitive:

\begin{equation}
\label{eq:outer1}
\{(R_1,R_2): R_1 \le C_1,
R_2 \le C_2 \}
\end{equation}
 
 Such an outer bound is clearly tight at at least two extreme points of the channel's capacity region. However, it is unclear if such a bound would be tight at other points of the region as no achievable strategy is known to exist. 
 
 To improve on the outer bound in (\ref{eq:outer1}), multiple approaches exist in literature. The most significant class of bounds is a genie-aided approach \cite{SKC08}, where a genie is provided to the transmitter, the receiver(s) or both. Although a powerful approach, this genie must be carefully designed so as to lead to computable expressions for the outer bound while not ``weakening" the outer bound significantly. 
 
 One example of this genie based approach is the so-called ``1-bit" genie \cite{ALG05}. Here, the transmitter is provided with a single bit of non-causal feedback per channel use, corresponding to:
 \[
 1(H \ge G) 
 \]

where $1(.)$ is the indicator function. This feedback transforms the channel into a set of parallel degraded channels, one corresponding to all settings when $H \ge G$ and another when $H < G$. As the channel capacity is known as a computable expression for each degraded channel, an ergodic average of the two rate-region forms an outer bound on the original channel in Section (\ref{sec:model}). It is shown in \cite{TYL08} that this bound is within 6 bits/channel use of the capacity region. Although this approach yields a rate region that is strictly better as an outer bound than the rate region in (\ref{eq:outer1}), it is still unclear if it is tight at points other than the two extreme points of the capacity region. 

This paper is based on building a class of computable outer bounds based on the K\"{o}rner-Marton for the broadcast channel \cite{M09}.  This outer bound, for a memoryless broadcast channel with channel transition probability $p(y_1,y_2|x)$ is given by:

\begin{equation}
\label{eq:km}
\{(R_1,R_2): \bigcup_{p(u,x)} R_1 \le I(X;Y_1|U) ~~ R_2 \le I(U;Y_2) \}
\end{equation}
Note that a similar expression, given by:

\begin{equation}
\label{eq:km2}
\{(R_1,R_2): \bigcup_{p(v,x)} R_1 \le I(V;Y_2) ~~ R_2 \le I(X;Y_2|V) \}
\end{equation}

is also a valid outer bound on this channel. In this paper, we will refer to either (\ref{eq:km}) or (\ref{eq:km2}) as the K\"{o}rner-Marton outer bound, with the context clarifying which permutations of users is under consideration.

The K\"{o}rner-Marton  outer bound can be generalized to a memoryless broadcast channel with state with transition probability $p(y_1,y_2|x,h,g)$, where the state is only known to the receivers as:

\begin{equation}
\label{eq:kmwithstate}
\{(R_1,R_2): \bigcup_{p(u,x)} R_1 \le I(X;Y_1|U,H) ~~ R_2 \le I(U;Y_2|G) \}
\end{equation}

Note that the K\"{o}rner-Marton outer bound is computable as an optimization problem for a discrete memoryless broadcast channel as a cardinality bound on the auxiliary random variable $U$ can be imposed \cite{AEA10}. However, in the case of the Gaussian fading broadcast channel, there is no cardinality bound on $U$ (or, for that matter, any restriction on the support of $p(u,x)$) and therefore the expression in (\ref{eq:kmwithstate}) is not directly computable. Thus, the rest of this chapter is dedicated to developing a framework under which the bound and its relaxations can be computed for the Gaussian fading broadcast channel.

  The K\"{o}rner-Marton  outer bounding framework is chosen for multiple reasons: first, it is known that this outer bound is tight for the degraded (static and fading) and MIMO Gaussian broadcast channels with full state knowledge at transmitter and receivers. Thus, it is conceivable that it would be a good choice of this channel as well. Second, the rate expressions involve only one auxiliary random variable, and therefore can be optimized with fewer steps than one with more parameters.

An arbitrary point  on the boundary of the K\"{o}rner-Marton outer bound in (\ref{eq:kmwithstate}) can be expressed in terms of the objective:

\eq{\label{KMouter} R_1+w R_2\le I(X;Y_1|U,H)+w I(U;Y_2|G),} 

where $w$ is a positive weighting factor. Intuitively, since the rate region is convex, $R_1 + w R_2$ represents a tangent to the rate-region boundary and all points are covered for $0 \le w \le \infty$. 

Our focus in this chapter is on a portion of the \korner-Marton rate region corresponding to $1 \le w\le \infty$ in (\ref{KMouter}).  Note that, despite this narrowing in the range of possible values for $w$, an upper bound corresponding to every point on the channel's rate region is obtained when a permuted optimization problem of the form:
 
 \eq{\label{KMouter2} {\hat w} R_1+ R_2\le {\hat w} I(V;Y_1|H)+  I(X;Y_2|V, G)} 

is considered for $1 \le {\hat w} \le \infty$. In other words, an intersection of the two outer bounds corresponding to (\ref{KMouter}) and (\ref{KMouter2}) will form an outer bound on the capacity region of the original channel. Given the symmetries and similarities between (\ref{KMouter}) and (\ref{KMouter2}), the rest of the chapter will focus on (\ref{KMouter}) alone, and leave the treatment of (\ref{KMouter2}) to the reader.

Next, we proceed to describe the main result of this paper.

\subsection{Main Result}

The main results of this paper is a computable outer-bound on a portion of capacity region of a class of fading broadcast channel. This outer-bound is tight  for a non-trivial class of channels. In Section \ref{sec:ach}, we address the proposed achievability scheme.  Throughout this section, we assume   $w\ge 1$ and  function $r(w,x)$ as defined as:

\eq{r(w,x) \triangleq \sum_{i=1}^n \frac{p_i}{x+\frac{1}{h_i^2}}- w \sum_{j=1}^m\frac{q_j}{x+\frac{1}{g_j^2}}.\label{rx}}

Let $1_k$ be a $k$-dimensional row vector of all one.
Next theorem, provides a condition under which we have a computable outerbound for the channel we introduced before.

\begin{thm}
\label{main:thm}
For each $w$, depending on the value of $r(w,x)$ for $0\le x \le Q$, we have the following cases:

{\bf Case 1: $r(w,Q^*)=0$:} In this case the following is an outer bound for the weighted sum-capacity $R_1+w R_2$:
\[R_1+wR_2\le \frac{1}{2}\Ex{\log(1+H^2Q^*)} + \frac{w}{2}\Ex{\log\left(\frac{1+G^2Q}{1+G^2Q^*}\right)},\]
if there exists a $n \times m$, positive matrix $A$ satisfying the following conditions:

\begin{itemize}
\item $ 1_n A=1_m$
\item $[a_1 a_2 \dots a_n] A = [b_1 b_2 \dots b_m]$
\item $\forall k \in \{1,2,\dots  n\}, \ \ \displaystyle{ \sum_{i=1}^k \frac{p_i}{a_i} \le w \sum_{j=1}^m \frac{q_j}{b_j} \left( \sum_{i=1}^k A_{ij}\right)}$
\end{itemize}
where $a_i=Q^*+\frac{1}{h_i^2}$ and $b_j=Q^*+\frac{1}{g_j^2}$.

The outerbound is tight if in addition to above conditions we have:
\[\Ex{ \log\left(\frac{ 1 + H^2 Q}{1+H^2 {Q}^*}\right)} \ge  \Ex{ \log\left(\frac{1 + G^2 Q}{1+G^2 {Q}^*}\right)}.\]
{\bf Case 2: $r(w,Q)>0$:} In this case the following is the weighted sum-capacity $R_1+w R_2$:
\[R_1+wR_2= \frac{1}{2}\Ex{\log(1+H^2Q)},\]
if there exists a $n \times m$, positive matrix $A$ satisfying the following conditions:

\begin{itemize}
\item $ 1_n A=1_m$
\item $[a_1 a_2 \dots a_n] A = [b_1 b_2 \dots b_m]$
\item $\forall k \in \{1,2,\dots  n\}, \ \ \displaystyle{ \sum_{i=k}^n \frac{p_i}{a_i} \ge w \sum_{j=1}^m \frac{q_j}{b_j} \left( \sum_{i=k}^n A_{ij}\right)}$
\end{itemize}
where $a_i=Q+\frac{1}{h_i^2}$ and $b_j=Q+\frac{1}{g_j^2}$.

{\bf Case 3: $r(w,0)<0$:} In this case the following is the weighted sum-capacity $R_1+w R_2$:
\[R_1+wR_2= \frac{w}{2}\Ex{\log(1+G^2Q)},\]
if there exists a $n \times m$, positive matrix $A$ satisfying the following conditions:

\begin{itemize}
\item $ 1_n A=1_m$
\item $[a_1 a_2 \dots a_n] A = [b_1 b_2 \dots b_m]$
\item $\forall k \in \{1,2,\dots  n\}, \ \ \displaystyle{ \sum_{i=1}^k \frac{p_i}{a_i} \le w \sum_{j=1}^m \frac{q_j}{b_j} \left( \sum_{i=1}^k A_{ij}\right)}$
\end{itemize}
where $a_i=\frac{1}{h_i^2}$ and $b_j=\frac{1}{g_j^2}$.

\end{thm}

The third condition in each case is referred as a majorization requirement.
\noindent{\em Proof:} See the appendix.

\begin{rem}
\label{main:rem1}
Note that we can always add ``virtual" fades to both the channels, i.e., fade coefficients with corresponding zero probability of occurence. This can actually help in finding a matrix $A$ in each case for the same problem.
\end{rem}

\begin{rem}
\label{main:rem2}
Checking the existence of such a matrix $A$ in theorem statement is equivalent to solving a Linear Program (LP). Therefore, one can efficiently check the conditions of Theorem \ref{main:thm}. Note that there is no systematic manner in which virtual fades can be introduced into the problem. However, for some classes of channels, this can be checked in a straightforward fashion.
\end{rem}

First, we show  Theorem \ref{main:thm} can be used for degraded channels.  The scheme we present for degraded channels can be generalized to a wider class of channels. For the degraded case, we must have$g_m<h_1$. For this case, let us first assume that there exists a $0\le Q^* \le Q$ such that $r(w,Q^*)=0$. Therefore, to evaluate the outer bound, we must satisfy the conditions corresponding to the first case in  Theorem \ref{main:thm} . Other cases can be handled similarly. We construct the following vectors:
\begin{align}
x&=[b_1 b_2 \dots b_m a_1 a_2 \dots a_n],\\
y&=[b_1 b_2 \dots b_m],\\
s&=[\underbrace{0 0 \dots 0}_{\textrm {m times}} \frac{p_1}{a_1} \frac{p_2}{a_2} \dots \frac{p_n}{a_n}],\\
t&=[w\frac{q_1}{b_1} w\frac{q_2}{b_2} \dots w\frac{q_m}{b_m}],
\end{align}
where $a_i=Q^*+\frac{1}{h_i^2}$ and $b_j=Q^*+\frac{1}{g_j^2}$. Note that we add virtual fades $g_1,g_2,\dots, g_m$ in order to construct $x$. 

To satisfy the conditions imposed in the first case within Theorem \ref{main:thm}, it is sufficient to show that there exists a non-negative $(m+n)\times m$ matrix $A$ satisfying the following conditions:
\begin{itemize}
\item $1_{m+n} A= 1_m$
\item $x A=y$
\item $\displaystyle{\forall k \in \{1,2,\dots ,m+n\},  \sum_{i=1}^k s_i \le \sum_{j=1}^m t_j \left(\sum_{i=1}^k A_{ij}\right)}$
\end{itemize}

It is easy to check that the following matrix $A$ satisfies all these three conditions:

\eq{A=\left[\begin{array}{c}
I_{m\times m}\\
0_{n\times m}
\end{array}\right].}

Next, we find that this can be generalized to non-degraded case as well. First, we note that the last condition above can be simplified to: 

\[\forall k \in \{1,2,\dots ,m+n\},  \sum_{i=1}^k s_i \le \sum_{j=1}^m t_j 1\left(x_i=y_j\right).\]
 

We define:

\eq{\label{cond:Q*}T \triangleq \sum_{i=1}^n \frac{p_i}{Q^*+\frac{1}{h_i^2}}= w \sum_{j=1}^m\frac{q_j}{Q^*+\frac{1}{g_j^2}}.}

Then we have the following theorem.
\begin{thm}
\label{thm1}
Consider a two-user Gaussian broadcast fading channel with channel coefficients $\{h_1<\dots<h_n\}$
and $\{g_1<\dots<g_m\}$ and corresponding probability distribution $p_i,q_j$. Assume $h_n>g_m$ and $p_n\neq 0$.

Define \[T_1\triangleq \frac{1}{1-p_n}\sum_{i=1}^{n-1} \frac{p_i}{Q^*+\frac{1}{h_i^2}}=\frac{T-\frac{p_n}{a_n}}{1-p_n},\]
and $g_0,g_{m+1}$ such that $0<g_0<h_1$ and $g_{m+1}>h_n$. Let
$k\in\{0,\dots,m+1\}$ be such that
\eq{\label{thm:cond1} \frac{1}{Q^*+ \frac{1}{g_k^2}}\le {T_1}\le \frac{1}{Q^*+ \frac{1}{g_{k+1}^2}}.}

Then the following is an upper-bound on the weighted sum-rate:

\[R_1+wR_2\le \frac{1}{2}\Ex{\log(1+H^2Q^*)} + \frac{w}{2}\Ex{\log\left(\frac{1+G^2Q}{1+G^2Q^*}\right)},\]

if 
\eq{\label{thm:cond2} (w-1)T_1 \ge w\sum_{j=1}^k q_j \left( T_1-\frac{1}{Q^*+\frac{1}{g_j^2}} \right).}
\end{thm}

\noindent{\em Proof:} See the appendix.

This theorem partially characterizes the capacity region of classes of Gaussian fading broadcast channels (this is explained in greater detail in the next section.)

\begin{cor}
\label{cor:onesided}
Consider  a non-degraded, one-sided fading broadcast channel, (i.e., a channel with $m=1$). Then, the following is an upper-bound on the weighted sum-rate:
\[R_1+wR_2\le \frac{1}{2}\Ex{\log(1+H^2Q^*)} + \frac{w}{2}\log\left(\frac{1+g^2Q}{1+g^2Q^*}\right),\]
\end{cor}
\noindent{\em Proof:}
It follows directly from Theorem \ref{thm1}. There are two cases:
\begin{enumerate}
\item $\displaystyle{\frac{1}{Q^*+\frac{1}{g^2}}>T_1} \Rightarrow k=0$
\item $\displaystyle{\frac{1}{Q^*+\frac{1}{g^2}}\le T_1} \Rightarrow k=1$
\end{enumerate}

In the first case, Condition (\ref{thm:cond2}) of Theorem \ref{thm1} simplifies to $(w-1)T_1\ge 0$ which comes from the fact that $w\ge 1$ and $T_1\ge 0$. In the second case, condition (\ref{thm:cond2}) holds because:
\[\frac{w}{Q^*+\frac{1}{g^2}}=T>T_1.\] \qed

We  show in the next section that this upper-bound is in fact tight and therefore part of the capacity region of one sided, fading broadcast channel, $w>1$, is characterized.

The following theorem also provide an upper-bound on the part of capacity region of another class of fading broadcast channels that are, in general, non-degraded.
\begin{thm}
\label{thm:deg-gen}
Consider a 2-user broadcast fading channel with channel fades  $\{h_1<\dots<h_n\}$
and $\{g_1<\dots<g_m\}$ with corresponding probability distribution $p_i,q_j$, where $h_n>g_m$ and $p_n\neq 0$. If channel fades satisfy one of the following two conditions:
\begin{enumerate}
\item $\displaystyle{(1-p_n) g_1^2 \ge \sum_{i=1}^{n-1} p_i h_i^2}$
\item $\displaystyle{\frac{1-p_n}{g_n^2} \ge \sum_{i=1}^{n-1} \frac{p_i}{h_i^2}}$
\end{enumerate}
Then the following is an upper-bound on the weighted sum-rate of the channel:
\[R_1+wR_2\le \frac{1}{2}\Ex{\log(1+H^2Q^*)} + \frac{w}{2}\Ex{\log\left(\frac{1+G^2Q}{1+G^2Q^*}\right)},\]

where $w\ge 1$, and $0\le Q^*\le Q$ is such that
\[T \triangleq \sum_{i=1}^n \frac{p_i}{Q^*+\frac{1}{h_i^2}}= w \sum_{j=1}^m\frac{q_j}{Q^*+\frac{1}{g_j^2}}.\]
\end{thm}

\noindent{\em Proof:} See the appendix.

Note that the degraded fading Gaussian broadcast channel is a particular example of a channel that satisfies the conditions imposed by Theorem \ref{thm:deg-gen}. We state this formally in the following corollary.
 
\begin{cor}
\label{cor:deg}
In a two user, fading broadcast channel, if $g_n<h_1$, the weighted-sum rate for $w>1$ is upper-bounded as follows:
\[R_1+wR_2\le  \frac{1}{2}\Ex{\log(1+H^2Q^*)} + \frac{w}{2}\Ex{\log\left(\frac{1+G^2Q}{1+G^2Q^*}\right)}.\] 
\end{cor}

\noindent{\em Proof:} 
If $g_n<h_1$, second condition of Theorem \ref{thm:deg-gen} holds and therefore the upper-bound is given by the corresponding theorem.
\qed

In the next section we present (simple) achievable strategies for this channel and show that, for some special cases, a portion of the capacity region  can be characterized.

\section{Achievable Scheme}
\label{sec:ach}
The achievable scheme we adopt is  Gaussian superposition coding. In essence, the transmitter superposes two codewords to derive its transmit sequence $X$, in the form
$X=X_1+X_2$. Here, $X_i$ is the codeword to be decoded at Receiver $i$.  Each codebook corresponding to $X_i$ is generated i.i.d. in accordance with a Gaussian distribution, with the variance of $X_1$ given by $0\le \tilde{Q}\le Q$ and of $X_2$ by $Q- \tilde{Q}$.  Given this, we have

\begin{align}
R_1 & \le  \frac{1}{2} \Ex{\log\left(1+H^2  \tilde{Q}\right)}\label{ach1}\\
R_2  & \le  \frac{1}{2}\min{\left(\Ex{ \log\left(\frac{ 1 + H^2 Q}{1+H^2 \tilde{Q}}\right)} ,  \Ex{ \log\left(\frac{1 + G^2 Q}{1+G^2 \tilde{Q}}\right)} \right)}\label{ach2}
\end{align}

as rates that are achievable in the system over all choices of $0\le\tilde{Q}\le Q$. In a similar spirit, the rates given by

\begin{equation}
\label{ach3}
\begin{array}{ll}
R_1 & \le \min{\left(\frac{1}{2} \Ex{ \log\left(\frac{ 1 + H^2 Q}{1+H^2  \tilde{Q}}\right)} , \frac{1}{2} \Ex{ \log\left(\frac{1 + G^2 Q}{1+G^2 \tilde{Q}}\right)} \right)} \\
R_2 & \le  \frac{1}{2} \Ex{ \log\left(1+G^2  \tilde{Q}\right)}
\end{array}
\end{equation}

are also achievable on this channel. The union of the two rate regions  given by (\ref{ach1}),(\ref{ach2}) and (\ref{ach3}) describes the entire achievable rate region for this channel.

In this section, we focus on achievable rates from Equations (\ref{ach1}) and (\ref{ach2}). Note that in Equation (\ref{ach2}), if 
\eq{\label{ach:cond}\Ex{ \log\left(\frac{ 1 + H^2 Q}{1+H^2 \tilde{Q}}\right)} \ge  \Ex{ \log\left(\frac{1 + G^2 Q}{1+G^2 \tilde{Q}}\right)},}

the following weighted sum-rate can be achieved:
\eq{\label{ach:wsr} R_1+wR_2\ge \frac{1}{2}\Ex{\log(1+H^2\tilde{Q})} + \frac{w}{2}\Ex{\log\left(\frac{1+G^2Q}{1+G^2\tilde{Q}}\right)}} 

We maximize right hand side of (\ref{ach:wsr}) with respect to $0 \le \tilde{Q} \le Q$, by taking derivative and set it to zero. The optimal $Q^*$ should satisfy the following equation:
\[\sum_{i=1}^n \frac{p_i}{Q^*+\frac{1}{h_i^2}}= w \sum_{j=1}^m\frac{q_j}{Q^*+\frac{1}{g_j^2}}\]
Note that this is the same condition as (\ref{cond:Q*}).
This means if for this optimal $Q^*$, condition (\ref{ach:cond}) and Condition (\ref{thm:cond2}) holds, the following weighted sum rate lies on the boundary of the capacity region of this channel for some $w>1$:

\eq{ R_1+wR_2= \frac{1}{2}\Ex{\log(1+H^2Q^*)} + \frac{w}{2}\Ex{\log\left(\frac{1+G^2Q}{1+G^2 Q^*}\right)}. \nonumber }

We use this fact to prove the following theorems.

\begin{thm}[A portion of the capacity region of one-sided fading broadcast channel]
\label{cap:onesided}
Consider a non-degraded, one sided fading broadcast channel, where the second channel is fixed. For $w\ge 1$, if Equation (\ref{cond:Q*}) has a solution in the range $0< Q^*<Q$,  the following characterizes  the weighted sum rate of this channel for $w$ 
\[
R_1+wR_2= \frac{1}{2}\Ex{\log(1+H^2Q^*)} + \frac{w}{2}\Ex{\log\left(\frac{1+g^2Q}{1+g^2 Q^*}\right)}.\]

If the solution of Equation (\ref{cond:Q*}) occurs outside the interval $[0,Q]$, weighted sum-capacity is achieved at one of the end points.
 
\end{thm}

\noindent{\em Proof:} See the appendix.

There are also other (non-one sided) channels  whose capacity region can be characterized partially. This  theorem presents one such class:

\begin{thm}
\label{cap:deg-gen}
Consider a broadcast fading channel with parameters satisfying the following:
\[\displaystyle{\frac{1-p_n}{g_n^2} \ge \sum_{i=1}^{n-1} \frac{p_i}{h_i^2}}.\]

  For $w\ge 1$, if Equation (\ref{cond:Q*}) has a solution in the range $0< Q^*<Q$,  the following weighted sum rate lies on the boundary of the capacity region of this channel:
\[
R_1+wR_2= \frac{1}{2}\Ex{\log(1+H^2Q^*)} + \frac{w}{2}\Ex{\log\left(\frac{1+G^2Q}{1+G^2 Q^*}\right)}.\]

 
\end{thm} 

\noindent{\em Proof:} See the appendix.


\bibliographystyle{plain} 
\bibliography{main_ref} 

\appendix

\begin{proof}[Proof of Theorem \ref{main:thm}]
 Consider again the optimization problem given by:


\eq{ R_1+w R_2\le I(X;Y_1|U,H)+w I(U;Y_2|G),} 
for all $w\ge 1$.

We can simplify the right  hand side of this inequality as:

\begin{align}
 I(X;Y_1|U,H)+&w I(U;Y_2|G) \notag\\
=&h(Y_1|U, H) - h(Y_1|X,U, H) + \notag\\
& w h(Y_2|G) - w h(Y_2|U, G) \notag\\
\le & h(Y_1|U, H) - w h(Y_2|U, G) + C, \label{Cupper}
\end{align}
where $C$ is a constant given by
\[
C \triangleq w \frac{1}{2} \Ex{ \log \left(2 \pi e \left(Q + \frac{ 1}{G^2}\right)\right)} - \frac{1}{2} \Ex{\log \left(2 \pi e\frac{1}{H^2}\right)}
\]

Note that this $C$ is an upper bound on $w h(Y_2|G) -  h(Y_1|X,U, H)$, using the fact that Gaussians maximize entropy given a second moment constraint.

Thus, from Equation (\ref{Cupper}), finding an outer bound for the weighted sum-rate $R_1+wR_2$ reduced to upper-bound the $h(Y_1|U,H)-wh(Y_2|U,G)$. We can further simplify this quantity as the following:
\begin{align}
h(Y_1|U,H)-wh(Y_2|U,G)=&\sum_{i=1}^n p_i h(Y_1|U,H=h_i)\notag\\
&-w\sum_{j=1}^m q_j h(Y_2|U,G=g_j) \label{Hwsr}
\end{align}

To simplify the math, throughout this section, we use the following notation:
\[a_i=Q^*+\frac{1}{h_i^2}, \forall i\in\{1,2,\dots,n\},\]
and
\[b_j=Q^*+\frac{1}{g_j^2}, \forall j\in\{1,2,\dots,m\}.\]

And therefore,
\[T\triangleq \sum_{i=1}^n\frac{p_i}{a_i}=w\sum_{j=1}^m \frac{q_j}{b_j}\]

In order to find an upperbound for (\ref{Hwsr}), we use a modified version of Costa's EPI as stated in the next lemma.

\begin{lem}
\label{lem:epi}
$f(t)=N\left(X+\sqrt{t}Z|U\right)=2^{2 h\left(X+\sqrt{t}Z|U\right)}$ is a concave function of $t$, where $Z$ is a Gaussian random variable independent of $X$ and $U$.
\end{lem}

Here we note that \eq{\label{hy2|u}h(Y_2|U,G=g_i)=h\left(X+\frac{Z}{g_i}\Big|U\right),} 
and therefore we can use Lemma \ref{lem:epi} and Jensen's inequality to upper-bound the right hand side of Equation (\ref{Hwsr}) as we address next.
As indicated in the Theorem \ref{thm1} statement, we assume the fades at each receiver are sorted by the index, i.e.:
\[h_1< h_2 < ...<h_n,\]
and 
\[g_1<g_2<...<g_m.\]


 
 
In order to use Lemma \ref{lem:epi} we need to define  matrix $A$, as the following:

\begin{defin}
\label{def:A}
Let $A$ be an $n\times m$ matrix of non-negative entries satisfying the first two conditions of  each cases of Theorem \ref{main:thm}.

Equivalently:
\begin{enumerate}
\item Sum of all the entries of each column of $A$ is one.
\item \[\left[\frac{1}{h_1^2} , \frac{1}{h_{2}^2} , \dots , \frac{1}{h_{n}^2}\right] A=\left[\frac{1}{g_1^2} , \frac{1}{g_2^2} , \dots , \frac{1}{g_m^2}\right].\]   
\end{enumerate}
\end{defin}

In fact, column $i$ of $A$ represents the convex coefficients of writing $g_j^{-2}$ in terms of $h_i^{-2}$'s.  Next, let \eq{\label{fi_assign} f_i=h(Y_1|U,H=h_i)=h\left(X+\frac{Z}{h_i}\Big|U\right).} 

From Lemma \ref{lem:epi}, and Jensen's inequality we can lower bound $h(Y_2|U,G=g_j)$ as the following:

\begin{align}
h(Y_2|U,G=g_j)&=h\left(X+\frac{Z}{g_i}\Bigg| U\right) \notag\\
&= \frac{1}{2}\log\left(N\left( X+\sqrt{\frac{1}{g_j^2}}Z\bigg|U\right)\right) \notag\\
&=\frac{1}{2}\log\left(N\left( X+\sqrt{\sum_{i=1}^n A_{ij} \frac{1}{h_i^2}}Z\bigg|U\right)\right) \label{a1}\\
&\ge \frac{1}{2}\log\left(\sum_{i=1}^n A_{ij} N\left( X+ \sqrt{\frac{1}{h_i^2}}Z\bigg|U\right)\right) \label{a2}\\
&= \frac{1}{2}\log\left(\sum_{i=1}^n A_{ij} 2^{2f_i}\right), \label{a3}
\end{align}
where (\ref{a1}) follows from Definition \ref{def:A}, (\ref{a2}) holds because of Jensen's inequality and (\ref{a3}) follows from the definition of $f_i$ in Equation (\ref{fi_assign}). Therefore, we lower bound $h(Y_2|U,G=g_j)$ as stated above for all $j$'s which gives a lower bound for Equation (\ref{Hwsr}) as the following:

\begin{align}
h(Y_1|U,F)-&wh(Y_2|U,F)\le \notag\\ &\sum_{i=1}^n p_i f_i - \frac{w}{2} \sum_{j=1}^m q_j \log\left(\sum_{i=1}^n A_{ij} 2^{2f_i}\right).\label{wsr_upper1}
\end{align}

%

Therefore any matrix $A$, satisfying conditions of Definition \ref{def:A}, gives an upper-bound on the capacity region of this channel. 

Next, we investigate to find the tightest bound in this class. Note that in the right hand side of Equation (\ref{wsr_upper1}), there are variables $f_i$ that we know corresponds to a valid entropy. We  capture this fact by imposing some constraints to $f_i$'s and instead removing the condition of $f_i$'s to be entropy functions\footnote{Condition of being in the form of Entropy function is a very sophisticated constraint and the main reason that this problem is still open!}. Next lemma gives one of such constraints:

\begin{lem}
\label{lem:epi2}
Consider $f_i$'s as defined in Equation (\ref{fi_assign}). Following relation holds for them:
\[\forall i<j \ : \ 2^{2f_i}-\frac{2\pi e}{h_i^2}\ge 2^{2f_j}-\frac{2\pi e}{h_j^2} \]  
\end{lem}
\begin{proof}
This follows from  EPI. Assume $i<j$, we know $\frac{1}{h_i^2}>\frac{1}{h_j^2}$. Thus, 
\begin{align*}
2^{2f_i}&=2^{2 h\left( X+\frac{Z}{h_i} \big| U \right)}=2^{2 h\left( X+\frac{Z}{h_j} +Z\sqrt{\frac{1}{h_i^2}-\frac{1}{h_j2}} \big| U \right)}\\
&\ge 2^{2 h\left( X+\frac{Z}{h_j} \big| U \right)} + 2^{2 h\left( Z\sqrt{\frac{1}{h_i^2}-\frac{1}{h_j^2}}\big| U \right)} \tag{a}\\
&= 2^{2f_j} + 2^{\log\left(2\pi e \left(\frac{1}{h_i^2} -\frac{1}{h_j^2}\right)\right)}\tag{b}\\
&=2^{2f_j} + \frac{2\pi e}{h_i^2} -\frac{2\pi e}{h_j^2},
\end{align*}
where (a) and (b) follows respectively from EPI, and independence of $U$ and $Z$. This concludes the proof of the Lemma.
\end{proof}

Next, we write the final optimization problem we would like to solve:

\begin{align}
\min_{A_{ij}} \max_{f_i} \ &\sum_{i=1}^{n} p_i f_i -\frac{w}{2} \sum_{j=1}^m q_j \log\left(\sum_{i=1}^n A_{ij} 2^{2f_i}\right) \label{opt}\\
s.t: \ \ & 1- \forall j, \  \sum_{i} A_{ij}=1 \notag\\
 & 2- \forall j, \ \sum_{i} \frac{A_{ij}}{h_i^2}=\frac{1}{g_j^2}\notag\\
& 3- \forall i, \ \frac{1}{2}\log\left( \frac{2 \pi e}{h_i^2}\right) \le f_i\le \frac{1}{2}\log\left( 2 \pi e \left(Q+\frac{1}{h_i^2}\right)\right)\notag\\
& 4-  \forall i<k  , \ 2^{2f_i}-\frac{2\pi e}{h_i^2}\ge 2^{2f_k}-\frac{2\pi e}{h_k^2}\notag
\end{align}

Note that, constraint 1 and 2 follows from Definition (\ref{def:A}), constraint 3 comes from the fact that entropy function is minimized if $X$ is deterministic and upper bounded by the entropy of a Gaussian random variable with the same second moment, and finally constraint 4 follows from Lemma (\ref{lem:epi}). It is also worth mentioning that constraint 3, can be obtained from the following two constraints in addition to constraint 4:
\[f_{n} \ge \frac{1}{2}\log\left( \frac{2 \pi e}{h_{n}^2}\right),\]
and
\[ f_1 \le  \frac{1}{2}\log\left( 2 \pi e \left(Q+\frac{1}{h_1^2}\right)\right).\]

Solving this min max problem gives the best outer-bound for the capacity region of the fading broadcast channel when $w>1$. Although solving this optimization problem is still not an easy task, we know we can find an outerbound for this capacity region by fixing an $A$ that satisfies the conditions 1 and 2 and solving the following problem instead:
 
\begin{align}
 \max_{f_i} \ &\sum_{i=1}^{n} p_i f_i -\frac{w}{2} \sum_{j=1}^m q_j \log\left(\sum_{i=1}^n A_{ij} 2^{2f_i}\right), \label{opt2}\\
s.t: \ \ 
& 1-f_{n} \ge \frac{1}{2}\log\left( \frac{2 \pi e}{h_{n}^2}\right) \notag\\
& 2-f_1 \le  \frac{1}{2}\log\left( 2 \pi e \left(Q+\frac{1}{h_1^2}\right)\right)\notag\\
& 3-  \forall i<k  , \ 2^{2f_i}-\frac{2\pi e}{h_i^2}\ge 2^{2f_k}-\frac{2\pi e}{h_k^2}\notag
\end{align}
when we know $A$ satisfies the conditions given by Definition \ref{def:A} or equivalently conditions 1 and 2 in the optimization problem given by (\ref{opt}).
In the next lemma we show this problem is convex with respect to $(f_1,f_2,\dots,f_n)$.

\begin{lem}
\label{lem:convex}
Objective function of optimization problem (\ref{opt2}) is convex.
\end{lem}
\begin{proof}[Proof of Lemma \ref{lem:convex}]
We show the Hessian of objective function is positive semi-definite.
Let $H$ be the Hessian. Let \[e_{lj}= \frac{\displaystyle A_{lj} 2^{2f_l}}{\displaystyle\sum_{i=1}^n A_{ij}2^{2 f_i}}\]
$H$ can be written as:

\[H_{lk}=\left\{ \begin{array}{cc}
-2w \displaystyle\sum_{j=1}^m q_j e_{lj} e_{kj} & l\neq k\\
\\
2w \displaystyle\sum_{j=1}^m q_j e_{lj} \sum_{i=1,i\neq l}^n e_{ij} & l=k
\end{array}\right.\]

Define the following matrices:

\[ME^{ij}_{lk}=\left\{ \begin{array}{cc}
\displaystyle \sum_{t=1,i\neq l}^n e_{tj} & l=k=i\\

\displaystyle  e_{lj} & l=k\neq i\\

- \displaystyle e_{lj} & l\neq k=i\\

- \displaystyle e_{kj} & l=i\neq k\\
\end{array}\right.\]

$H$ can be written as the sum of $ME^{ij}$'s as the following:

\eq{\label{H:ME}
H= w \sum_{j=1}^m q_j\sum_{i=1}^n e_{ij} ME^{ij}.
}

Since $w, q_j$ and $e_{ij}$ are positive, in order to show that $H$ is positive semi definite, it is enough to show $ME^{ij}$ is positive semi definite. Let $x$ be a vector of size $n$, we show $x^T ME^{ij} x \ge 0$. This quantity simplifies as:
\[x^T ME^{ij} x = \sum_{t=1, t\neq i}^n e_{tj} (x_{t}-x_{i})^2 \ge 0.\]

This completes the proof.
\end{proof}

We prove the first case of Theorem \ref{main:thm} the other two cases can be proved in a similar way. We show $f_i=\frac{1}{2}\log(2\pi e(Q^*+\frac{1}{h_i^2}))$ is the optimal allocation for optimization problem (\ref{opt2}),  if we use matrix $A$ satisfies the conditions given in Theorem \ref{main:thm}. 
For problem (\ref{opt2}), we write the Lagrangian. Let the {\it positive} dual variables be $\mu$ for the first constraint, $\eta$ for the second constraint and $\lambda_1,\dots,\lambda_{n-1}$ for the third condition\footnote{Note that in the condition 2, it is enough to consider $k=i+1$ for $i=1,\dots,n-1$}. 
 KKT conditions for the optimal allocation are as the following:
\begin{align}
i=1 \ : & \ p_1-\sum_{j=1}^m\frac{\displaystyle{w q_j A_{1j} 2^{2f_1}}}{\displaystyle{\sum_{i=1}^n A_{ij} 2^{2f_i}}} - \mu +\lambda_1 2^{2f_1}=0 \label{cond1}\\
i=n \ : & \ p_n-\sum_{j=1}^m\frac{\displaystyle{w q_j A_{nj} 2^{2f_n}}}{\displaystyle{\sum_{i=0}^n A_{ij} 2^{2f_i}}} +\eta -\lambda_{n-1}2^{2f_n}=0\label{cond2}\\
1 < i<n \ : & \ p_i-\sum_{j=1}^m\frac{\displaystyle{w q_j A_{ij} 2^{2f_i}}}{\displaystyle{\sum_{i=1}^n A_{ij} 2^{2f_i}}} + (\lambda_i-\lambda_{i-1}) 2^{2f_i}=0\label{cond3}\\
& \ \eta \left(f_n- \frac{1}{2}\log\left( \frac{2 \pi e}{h_n^2}\right)\right)=0\label{cond4}\\
& \ \mu \left(\frac{1}{2}\log\left( 2 \pi e \left(Q+\frac{1}{h_1^2}\right)\right)-f_1\right)=0\label{cond5}\\
 1< i< n \ : & \ \lambda_{i-1}\left( 2^{2f_{i-1}}-\frac{2\pi e}{h_{i-1}^2}-2^{2f_i}+\frac{2\pi e}{h_i^2} \right)=0\label{cond6}
\end{align}

Next, we show the following satisfies these constraint and therefore is the optimal solution for this optimization problem:
\eq{ f_i=\frac{1}{2}\log\left(2\pi e\left(Q^*+\frac{1}{h_i^2}\right)\right)\label{assign1}.}

Note that with these $f_i$'s, Equation (\ref{cond6}), immediately follows. From Equations  (\ref{cond4}, \ref{cond5}), and definition of $f_i$'s as given in (\ref{assign1}), $\mu=\eta=0$.

From the conditions on matrix $A$ in Theorem \ref{main:thm} and Equation (\ref{assign1}), it is easy to obtain: \eq{\label{p1} \sum_{i=1}^n A_{ij} 2^{2f_i}=Q^*+\frac{1}{g_j^2},}
therefore, Equation (\ref{cond1}) is equivalent to show:
\eq{\label{p2}w \sum_{j=1}^m\frac{A_{1j}q_j}{Q^*+\frac{1}{g_j^2}}\ge  \frac{p_1}{Q^*+\frac{1}{h_1^2}}.}

%

Conditions given by Equation (\ref{cond3}), can be simplified as above to get:
\[\sum_{i=1}^k \frac{p_i}{Q^*+\frac{1}{h_i^2}} \le w \sum_{j=1}^m\left(\frac{q_j}{Q^*+\frac{1}{g_j^2}}\sum_{i=1}^k A_{ij}\right).\]

These conditions holds from The third conditions of each case in Theorem \ref{main:thm}.

Finally Equation (\ref{cond2}), is simplified as:
\[\sum_{i=1}^n \frac{p_i}{Q^*+\frac{1}{h_i^2}} = w \sum_{j=1}^m\frac{q_j}{Q^*+\frac{1}{g_j^2}},\]
that follows from the assumption of the first case in the theorem. Thus, all the KKT conditions are satisfied and therefore $f_i$'s as defined in Equation (\ref{assign1}) is the solution for the Optimization problem (\ref{opt2}).
\end{proof}

\begin{proof}[Proof of Theorem \ref{thm1}]
We construct a matrix $A$ satisfying the conditions of the first case of Theorem \ref{main:thm}. We add a virtual fade $h_0$ with zero probability of occurrence. Later, we let it go to zero. 
\

Consider the following $(n+1)\times m$ matrix $A$:
\eq{\label{matrix:A}
\hspace{-.15in}A^T=\left[\begin{array}{cccccc}
\beta_1 & \alpha_1 \frac{p_1 b_1}{wq_1 a_1} & \alpha_1 \frac{p_2 b_1}{wq_1 a_2} & \dots & \alpha_1 \frac{p_{n-1} b_1}{wq_1a_{n-1}} & \gamma_1 \\
\beta_2 & \alpha_2 \frac{p_1b_2}{wq_2a_1} & \alpha_2 \frac{p_2b_2}{wq_2a_2} & \dots & \alpha_2 \frac{p_{n-1}b_2}{wq_2a_{n-1}} & \gamma_2 \\
\vdots & \vdots & \vdots& \ddots & \vdots &\vdots\\
\beta_m & \alpha_m \frac{p_1b_m}{wq_ma_1} & \alpha_2 \frac{p_2b_m}{wq_ma_2} & \dots & \alpha_2 \frac{p_{n-1}b_m}{wq_ma_{n-1}} & \gamma_m 
\end{array}\right]}
where:
\eq{\label{betam}
\beta_j=\frac{b_j-a_n+\frac{\alpha_j b_j}{w q_j}(a_nT-1)}{a_0-a_n},}
and
\eq{\label{gammam}
\gamma_j=\frac{a_0-b_j+\frac{\alpha_jb_j}{wq_j}\left(1-p_n-a_0T+\frac{p_n a_0}{a_n}\right)}{a_0-a_n}.}

Note that with these choice of $\beta_j$ and $\gamma_j$, for all $\alpha_j$'s, matrix $A$ satisfies all the conditions of Theorem \ref{main:thm}.
In the next lemma, we show existence of $\alpha_j$'s that make all the entries of matrix $A$  positive.

\begin{lem}
\label{existence:alphas}
If Condition (\ref{thm:cond2}) holds, there exists positive $\alpha_j$'s such that $\sum_{j=1}^m \alpha_j=1$ and the matrix $A$ as defined in Equation (\ref{matrix:A}) has positive entries.
\end{lem}
\begin{proof}[Proof of Lemma \ref{existence:alphas}]
In order to make the entries of $A$ positive, it is enough if $\alpha_j,\beta_j, \gamma_j>0$. Consider positive $\alpha_j$'s, we first find the conditions on $\alpha_j$'s that makes $\gamma_j$'s positive. We can rewrite $\gamma_j$ as
\[\gamma_j=\frac{a_0\left(1-\frac{\alpha_j b_j}{w q_j}(T-\frac{p_n}{a_n})\right)+ \frac{\alpha_j b_j}{w q_j}(1-p_n)-b_j}{a_0-a_n}.\]

Since $a_0$ corresponds to a dummy fade therefore we can increase it $a_0\rightarrow \infty$ (decrease $h_0$ to zero). Thus we can make $\gamma_j$ positive if coefficient of $a_0$ is positive in the numerator. It gives the following constraint on $\alpha_j$'s:
\eq{\label{alpha:cond1} \alpha_j \le \frac{w q_j}{b_j} \frac{1}{T-\frac{p_n}{a_n}}.}

Next, we derive conditions on $\alpha_j$'s that makes $\beta_j$'s positive. We also refer to entries of $A$ as $A_{ij}$, where $i \in \{0,1,\dots ,n\}$ and $j\in\{1,2,\dots,m\}$. 
It is easy to check that under the following condition, $\beta_j$ is positive:

\eq{\label{alpha:cond2} \alpha_j \le \frac{w q_j}{b_j}\frac{b_j-a_n}{1- a_n T}.}

Thus, if $\alpha_j$ satisfies both (\ref{alpha:cond1}) and (\ref{alpha:cond2}), $\gamma_j$ and $\beta_j$ are positive. We can simplify these two conditions and write them as the following:

\[ \alpha_j \le \frac{w q_j}{b_j} \min\left\{\frac{b_j-a_n}{1-a_n T},  \frac{1}{T-\frac{p_n}{a_n}} \right\}. \]
Therefore, the lemma is proved if we can show the followings:

\eq{\label{alpha:cond} \sum_{j=1}^m \frac{w q_j}{b_j} \min\left\{\frac{b_j-a_n}{1-a_n T},  \frac{1}{T-\frac{p_n}{a_n}} \right\} \ge 1.}

It is easy to check that $\frac{b_j-a_n}{1- a_nT}\le \frac{1}{T-\frac{p_n}{a_n}}$ if and only if $b_j\le \frac{1}{T_1}$. Let $k$ be defined as the index of the smallest $b_j$, greater than $\frac{1}{T_1}$, as it is defined in Theorem \ref{thm1}. Condition (\ref{alpha:cond}) can be simplified as in Equation (\ref{lem:minsum}).

\begin{table*}
\label{table1}
\begin{align}
\frac{w q_j}{b_j} \min\left\{\frac{b_j-a_n}{1-a_n T},  \frac{1}{T-\frac{p_n}{a_n}} \right\} &=
w \sum_{j=1}^k \frac{q_j}{b_j} \frac{1}{T-\frac{p_n}{a_n}}+w \sum_{j=k+1}^m \frac{q_j}{b_j}\frac{b_j-a_n}{1-a_nT}\notag\\
&=\frac{w}{T-\frac{p_n}{a_n}} \sum_{j=1}^k \frac{q_j}{b_j} + \frac{w}{1-a_nT}\left(1-\sum_{j=1}^k q_j\right) -\frac{a_n}{1-a_nT}\left(T-w\sum_{j=1}^k\frac{q_j}{b_j}\right)\notag\\
&=\frac{1}{1-a_nT}\left(w\sum_{j=1}^k q_j\left(\frac{1}{T_1 b_j-1}\right)+w-a_nT\right)\label{lem:minsum}
\end{align}
\end{table*}

Replacing (\ref{lem:minsum}) with the left hand side of (\ref{alpha:cond}), the condition simplifies to:

\eq{\label{lem:end} (w-1)T_1 \ge w\sum_{j=1}^k q_j \left( T_1-\frac{1}{Q^*+\frac{1}{g_j^2}} \right),}

which is equivalent to Equation (\ref{thm:cond2}) and holds true from assumption of lemma. This proves the lemma.
\end{proof}

Note that from Definition of $A$ in (\ref{matrix:A}), we have:

\begin{align}
w \sum_{j=1}^m\left(\frac{q_j}{Q^*+\frac{1}{g_j^2}}\sum_{i=0}^k A_{ij}\right)=& w\sum_{j=1}^m\frac{\beta_j q_j}{b_j} + \notag\\
&w\sum_{j=1}^m\left(\frac{q_j}{b_j}\sum_{i=1}^k \frac{\alpha_j p_i b_j}{w q_j a_i}\right) \notag\\
=& w\sum_{j=1}^m\frac{\beta_j q_j}{b_j} + \sum_{j=1}^m\left(\alpha_j \sum_{i=1}^k \frac{ p_i }{ a_i}\right) \notag\\
=&w\sum_{j=1}^m\frac{\beta_j q_j}{b_j} + \notag\\
&\left(\sum_{j=1}^m\alpha_j\right) \left(\sum_{i=1}^k \frac{ p_i }{ a_i}\right) \notag\\
=&w\sum_{j=1}^m\frac{\beta_j q_j}{b_j} + \sum_{i=1}^k \frac{ p_i }{ a_i} \label{alpha:sum}\\
\ge& \sum_{i=1}^k \frac{ p_i }{ a_i}, \notag
\end{align}
 where Equation (\ref{alpha:sum}) follows from Lemma \ref{existence:alphas} and the fact that sum of $\alpha_j$'s are one.
Therefore, $A$ satisfies all the conditions of the first case in Theorem \ref{main:thm}. This proofs the theorem.
\end{proof}

\begin{proof}[Proof of Theorem \ref{thm:deg-gen}]
First, assume  condition 1 holds. So, 
\eq{\label{thm2:cond1} \displaystyle{(1-p_n) g_1^2 \ge \sum_{i=1}^{n-1} p_i h_i^2}.}

Let 
\[f(x)=\frac{1}{1-p_n}\sum_{i=1}^{n-1}\frac{p_i}{x+\frac{1}{h_i^2}}-\frac{1}{x+\frac{1}{g_1^2}}.\]

Note that:
\eq{\label{thm2:f0} f(0)=\frac{\sum_{i=1}^{n-1} p_i h_i^2}{1-p_n}-g_1^2\le 0,}
where inequality follows from the Equation (\ref{thm2:cond1}).
Taking the derivative of $f(x)$, and utilizing Jensen's inequality we can write:
\begin{align*}
f'(x)=&-\frac{1}{1-p_n}\sum_{i=1}^{n-1}\frac{p_i}{\left(x+\frac{1}{h_i^2}\right)^2}+\frac{1}{\left(x+\frac{1}{g_1^2}\right)^2}\\
\le& -\frac{1}{1-p_n}\left(\sum_{i=1}^{n-1}\frac{p_i}{x+\frac{1}{h_i^2}^2}\right)^2+\left(\frac{1}{x+\frac{1}{g_1^2}}\right)^2\\
=&\left(-\frac{1}{1-p_n}\sum_{i=1}^{n-1}\frac{p_i}{x+\frac{1}{h_i^2}}+\frac{1}{x+\frac{1}{g_1^2}}\right)\times \\
&\left(\frac{1}{1-p_n}\sum_{i=1}^{n-1}\frac{p_i}{x+\frac{1}{h_i^2}}+\frac{1}{x+\frac{1}{g_1^2}}\right)\\
=&-f(x)*V,
\end{align*}
where $V>0$ for all positive $x$. A simple argument show that this condition, $f'(x)\le -f(x) V$,  guarantee that $f(y)<0$ for all $y>x_0$ if $f(x_0)<0$. Since $f(0)<0$, $f(x)$ is also negative for all $x\ge 0$. Choosing $x=Q^*$, $f(Q^*)=T_1-\frac{1}{Q^*+\frac{1}{g_1^2}}\le 0$, and therefore:
\[T_1 \le \frac{1}{Q^*+\frac{1}{g_1^2}}.\]

Now in order to satisfy Condition \ref{thm:cond2} of Theorem \ref{thm1}, it is enough to show $(w-1)T_1\ge 0$ which is always the case.

To prove the second part of the theorem and show upper-bound holds under the condition 2, we write:

\begin{align}
T_1&=\frac{1}{1-p_n}\sum_{i=1}^{n-1}\frac{p_i}{\displaystyle Q^*+\frac{1}{h_i^2}} \notag\\
&\ge \frac{1}{Q^*+\displaystyle{\frac{1}{1-p_n}\sum_{i=1}^{n-1}\frac{p_i}{h_i^2}}} \label{thm2:a}\\
&\ge \frac{1}{Q^*+\displaystyle{\frac{1}{g_n^2}}} \label{thm2:b},
\end{align}
where (\ref{thm2:a}) and (\ref{thm2:b}) follows from Jensen's inequality and the second condition of Theorem respectively. Next, to check  the Condition \ref{thm:cond2} of Theorem \ref{thm1}, we need to show 
\eq{\label{thm2:t1}(w-1)T_1\ge w \sum_{j=1}^m q_i\left(T_1-\frac{1}{Q^*+\frac{1}{g_j^2}}\right).}

Note that the right hand side of (\ref{thm2:t1}) is equal to $w T_1-T$ and therefore (\ref{thm2:t1}), simplifies to show $(w-1)T_1\ge wT_1 -T$, or in other words, $T>T_1$. It is easy to check that $T>T_1$ since $h_n>h_j$ for all $j$'s.
\end{proof}

\begin{proof}[Proof of Theorem \ref{cap:onesided}]

Assume there exists $0<Q^*<Q$ satisfying Equation \ref{cond:Q*}. We prove  $Q^*$
satisfies condition (\ref{ach:cond}). From Jensen's inequality we can write:
\begin{align*} 
\Ex{ \log\left(\frac{ 1 + H^2 Q}{1+H^2 Q^*}\right)} &\ge \log \Ex{\frac{ 1 + H^2 Q}{1+H^2 Q^*}}\\
&=\log\left( 1+ \Ex{\frac{Q-Q^*}{Q^*+\frac{1}{H^2}}}\right)\\
&=\log\left(1+w \frac{Q-Q^*}{Q^*+\frac{1}{g^2}}\right)\\
& \ge \log\left(\frac{ 1 + g^2 Q}{1+g^2 Q^*}\right)
\end{align*}
Thus, the desired result follows from Equation (\ref{ach:wsr}) and Corollary \ref{cor:onesided}. 
\end{proof}

\begin{proof}[Proof of Theorem \ref{cap:deg-gen}]
From the fact that $h_n\ge g_m$ and using theorem assumption, it is easy to show:
\eq{\label{thmcapgen} \frac{1}{g_n^2} \ge \sum_{i=1}^n\frac{p_i}{h_i^2}=\Ex{\frac{1}{H^2}}.}
Condition (\ref{ach:cond}) follows from:

\begin{align} 
\Ex{ \log\left(\frac{ 1 + H^2 Q}{1+H^2 Q^*}\right)} &\ge \log \Ex{\frac{ 1 + H^2 Q}{1+H^2 Q^*}}\label{thm4:a}\\
&=\log\left( 1+ \Ex{\frac{Q-Q^*}{Q^*+\frac{1}{H^2}}}\right)\notag\\
&\ge \log\left(1+ \frac{Q-Q^*}{Q^*+\Ex{\frac{1}{H^2}}}\right)\label{thm4:b}\\
& \ge \log\left(\frac{ 1 + g_n^2 Q}{1+g_n^2 Q^*}\right)\label{thm4:c}\\
& \ge \Ex{\log\left(\frac{ 1 + G^2 Q}{1+G^2 Q^*}\right)},\label{thm4:d}
\end{align}
where (\ref{thm4:a}) and (\ref{thm4:b}) follows from Jensen's inequality, (\ref{thm4:c}) follows from Equation (\ref{thmcapgen}) and (\ref{thm4:d}) holds because $g_m>g_j$ for all $j$'s. Thus, the desired result follows from Equation (\ref{ach:wsr}) and Theorem \ref{thm:deg-gen}. 
\end{proof}       
\end{document}